\documentclass{article} 
\usepackage{arxiv}
\usepackage{amsthm}

\usepackage{natbib}        % required for bibliography
\usepackage{hyperref}

\usepackage{graphicx}          % Include this line if your
\usepackage{tikz-cd}
\usepackage{subfig}

\usepackage{xcolor}  % Needed for colours - load before mdframed

\usepackage{amssymb}
\usepackage{amsmath}
\usepackage{mathdots}

\usepackage{mathtools}

\usepackage{tensor}
\usepackage{xifthen}
\usepackage{urwchancal}
\DeclareFontFamily{OT1}{pzc}{}
\DeclareFontShape{OT1}{pzc}{m}{it}{<-> s * [1.10] pzcmi7t}{}
\DeclareMathAlphabet{\mathpzc}{OT1}{pzc}{m}{it}

\usepackage{graphicx}
\usepackage{ifpdf}
\ifpdf
\usepackage{epstopdf}
\PrependGraphicsExtensions{.svg}
\fi
\newtheorem{theorem}{Theorem}[section]

\newtheorem{remark}[theorem]{Remark}

\providecommand{\R}{\mathbb{R}}
\providecommand{\SO}{\mathbf{SO}}

\providecommand{\SE}{\mathbf{SE}}
\providecommand{\grpG}{\mathbf{G}}

\providecommand{\gothg}{\mathfrak{g}}

\providecommand{\so}{\mathfrak{so}}
\providecommand{\se}{\mathfrak{se}}

\providecommand{\calB}{\mathcal{B}}

\providecommand{\calL}{\mathcal{L}}

\providecommand{\Id}{I} % identity of a matrix group.

\providecommand{\tR}{\mathrm{R}} % left multiplication

\DeclareMathOperator{\tr}{tr}

\DeclareMathOperator{\diag}{diag}

\DeclareMathOperator{\Ad}{Ad}
\DeclareMathOperator{\ad}{ad}
\DeclareFontEncoding{LS1}{}{}
\DeclareFontSubstitution{LS1}{stix}{m}{n}
\DeclareSymbolFont{stixletters}{LS1}{stix}{m}{it}
\DeclareMathAccent{\cev}{\mathord}{stixletters}{"91}
\DeclareMathAccent{\vec}{\mathord}{stixletters}{"92}
\DeclareMathAccent{\vecev}{\mathord}{stixletters}{"95}

\providecommand{\td}{\mathrm{d}}

\providecommand{\ddt}{\frac{\td}{\td t}}

\usepackage{accents}
\usepackage{mathtools}
\makeatletter
\providecommand{\scirc}{%
    \hbox{\fontfamily{\rmdefault}\fontsize{0.4\dimexpr(\f@size pt)}{0}\selectfont{\raisebox{-0.52ex}[0ex][-0.52ex]{$\circ$}}}}
\makeatother
\makeatletter
\providecommand{\ucirc}{%
    \hbox{\fontfamily{\rmdefault}\fontsize{0.4\dimexpr(\f@size pt)}{0}\selectfont{\raisebox{0.0ex}[0ex][-0.52ex]{$\circ$}}}}

\makeatother

\mathchardef\mhyphen="2D
\providecommand{\idx}[5][]{
\ifthenelse{\isempty{#1}}% nothing is in the superscript spot
{\tensor*[_{#4}^{#3}]{#2}{_{#5}}}% if condition: old \idx
{\tensor*[_{#4}^{#3}]{#2}{^{#1}_{#5}}}% else condition old \ids
}
\usepackage{url}

\newcommand{\pp}[2]{ \frac{\partial #1}{\partial #2}}

\newcommand{\inert}{\mathbb{I}}
\newcommand{\binert}{\bar{\inert}_t}

\newcommand{\publicationdetails}
{
\copyrightNoticeIFACAccepted{2026} % Use for IFAC papers
M. Hampsey, P. van Goor, R. Banavar and R. Mahony, ``Tracking Control for a Dynamic Model of an Underwater Submersible'', in Proceedings of IFAC World Congress, 2026.\\
\\
}
\newcommand{\publicationversion}
\begin{document}

\title{Tracking Control for a Dynamic Model of an Underwater Submersible}
\headertitle{Tracking Control for a Dynamic Model of an Underwater Submersible}


% \thanks[footnoteinfo]{\todo{Sponsor and financial support acknowledgment
% goes here. Paper titles should be written in uppercase and lowercase
% letters, not all uppercase.}}
\author{
\href{}{\includegraphics[scale=0.06]{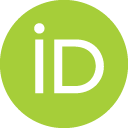}\hspace{1mm}
Matthew Hampsey}
\\
	Systems Theory and Robotics Group \\
	Australian National University \\
    ACT, 2601, Australia \\
	\texttt{matthew.hampsey@anu.edu.au} \\
	\And	\href{https://orcid.org/0000-0003-4391-7014}{\includegraphics[scale=0.06]{orcid.png}\hspace{1mm}
Pieter van Goor}
\\
	School of Aerospace, Mechanical,\\
 and Mechatronic Engineering,\\
 The University of Sydney, Australia.\\
    \texttt{pieter.vangoor@sydney.edu.au} \\
\And	\href{https://orcid.org/0000-0002-5746-7096}{\includegraphics[scale=0.06]{orcid.png}\hspace{1mm}
    Ravi Banavar}
\\
Centre for Systems and Control,\\
Indian Institute of Technology Bombay, India\\
\texttt{banavar@iitb.ac.in}
\And	\href{https://orcid.org/0000-0002-7803-2868}{\includegraphics[scale=0.06]{orcid.png}\hspace{1mm}
Robert Mahony}
\\
	Systems Theory and Robotics Group \\
	Australian National University \\
	ACT, 2601, Australia \\
	\texttt{robert.mahony@anu.edu.au} \\
}

\maketitle

\title{Tracking Control for a Dynamic Model of an Underwater Submersible\thanksref{footnoteinfo}}

\begin{abstract}                % Abstract of 50--100 words
Underwater vehicles are naturally modelled as rigid bodies on $\SE(3)$ subjected to added mass effects.
The passivity of the Hamiltonian structure of the system can be exploited to design energy-based stabilising controllers, however, the extension of these control designs to tracking control is not trivial since the 
error system for the classical error formulations is not itself Hamiltonian. 
In this paper, we show that a novel choice of error function leads to error dynamics that are Hamiltonian. 
We go on to derive an energy-based tracking control for a fully coupled model of a submersible vehicle. 
Asymptotic convergence of the control scheme is proved and the control is demonstrated in a simulation study of the Blue Robotics BlueROV2 Heavy submersible.
\end{abstract}

%===============================================================================

\section{Introduction}
The stabilisation of a dynamical system to a predefined trajectory is a fundamental task in control of mechanical systems.
For aerial robots like multirotors, the tracking task is simplified by the decoupling of rotational and translational inertias that allows the system to be modelled on $\SO(3) \times \R^3$ and allows attitude and position dynamics to be controlled in a cascaded manner \citep{leeGeometricTrackingControl2010}.
Submersible robots, on the other hand, are subject to added mass effects and gravity-induced moments due to the buoyancy and inertia of the surrounding volume of water \citep{leonardStabilityBottomheavyUnderwater1997}.
These dynamic effects couple the rotational and translational inertias and prevent a decoupled cascade control design. 

One classical approach to trajectory tracking for mechanical systems is that of computed torque \citep{slotine1991applied}. 
This approach uses feed-forward compensation (computed torque) to assign Euler-Lagrange dynamics to an error system expressed in generalised variables and then applies PD control to regulate the error. 
The design approach has been generalised to systems on Lie groups by a number of authors 
\citep{bulloGeometricControlMechanical2005,maithripalaAlmostglobalTrackingSimple2006,chandrasekaranGeometricPIDController2024}. 
Passivity-based control \citep{Ortega1998} exploits structure of the Euler-Lagrange system to shape the energy to achieve the desired stability properties of the error system while avoiding direct cancellation of nonlinear terms in the error system (present in computed torque control).  
A related technique applicable to port-Hamiltonian systems is interconnection and damping assignment passivity-based control (IDA-PBC) \citep{ortegaEnergyshapingPortcontrolledHamiltonian1999}, whereby the interconnection matrix, damping matrix and total energy of a port-Hamiltonian system are modified with feedback into a desired form.
\cite{blochControlledLagrangiansStabilization2000} proposed the method of controlled Lagrangians to modify the kinetic energy of an Euler-Lagrange system.
This approach was recently applied to the stabilisation of the submersible by \cite{contrerasControlledLagrangiansStabilization2022}. 
A limitation of all these energy based control approaches is that only state regulation or stabilisation to constant energy trajectories was achievable within the passive framework.  
\cite{fujimotoTrajectoryTrackingControl2003} considered general trajectory tracking for port-Hamiltonian systems by constructing a (generalised) canonical transformation, often found by solving an associated partial differential equation, that allows the definition of an error system that is itself Hamiltonian.
\cite{mahonyNovelPassivitybasedTrajectory2019} included the energy of the desired trajectory in the stabilisation control design and regulated the zero dynamics of the error variable. 

A parallel research theme in the literature developed tools to exploit the Lie group structure of mechanical systems with symmetry for observer and filter design 
\citep{mahonyNonlinearComplementaryFilters2008,bonnabelSymmetrypreservingObservers2008,barrauInvariantExtendedKalman2017}. 
Observer design is closely related to tracking control and there are a number of recent works that exploit this connection \citep{hampseyExploitingEquivarianceDesign2024,hampseySpatialGroupError2024,hampseyTrackingControlHomogeneous2023,weldeAlmostGlobalAsymptotic2024}.
In \cite{hampseyEquivariantTrackingControl2025}, the authors used a semi-direct group structure on the phase space $\grpG \times \gothg^\ast$ and showed that the Lie-Poisson equations exhibit equivariant symmetry. 
This allows the application of equivariant design techniques to dynamical mechanical systems with symmetry and the authors applied this to a trajectory tracking problem on $\SO(3)$ \citep{hampseyEquivariantTrackingControl2025} where the error system was shown to be (time varying) Lie-Poisson.

In this paper, we propose a tracking control scheme for fully-actuated underwater submersibles with coincident centres of pressure and mass.
This extends the example systems proposed in \cite{hampseyEquivariantTrackingControl2025} to systems with coupled inertias and serves as a practical guide on the use of the proposed constructions.
We review classical $\SE(3)$ Lie group theory in \S\ref{sec:prelims} and go on to present the dynamics of submersibles on $\SE(3)$, including added mass effects, in \S\ref{sec:dynamics_sub}. 
In \S\ref{sec:lie_group_error}, we formulate the semi-direct Lie group product on the phase space $\SE(3) \times \R^6$ and use this to derive an equivariant error for tracking control.
The error system is shown to be a (time-varying) Lie-Poisson system and this allows us to go on to propose a geometrically motivated PD tracking control for the underwater submersible in \S\ref{sec:tracking_control}.
We prove asymptotic tracking by showing that the equivariant error is stable using a passivity based argument. 
This proof simplifies the general proof proposed in \cite{hampseyEquivariantTrackingControl2025} by using the compactness of the $\SO(3)$ component of $\SE(3)$. 
Finally, in \S\ref{sec:simulation} and \S\ref{sec:discussion}, we simulate the Blue Robotics BlueROV2 Heavy submersible tracking a realistic trajectory, verifying the control performance.

\section{preliminaries}
\label{sec:prelims}

The Lie group $\SE(3)$ is the matrix Lie group defined by
\begin{align*}
    \SE(3) = \left\{ \begin{pmatrix}
        R & x \\ 0 & 1
    \end{pmatrix} \in \R^{4 \times 4}: R \in \SO(3), x \in \R^3 \right\},
\end{align*}
where $\SO(3)$ is the rotation group:
\begin{align*}
    \SO(3) = \left\{ R \in \R^{3 \times 3} : R^\top R = \Id, \det(R) = 1 \right\}.
\end{align*}
Note that $\SE(3) \cong \SO(3) \times \R^3$ as manifolds, and so for conciseness we will often write $(R, x) \in \SE(3)$ unambiguously.
The associated Lie algebra $\se(3)$ is the vector space
\begin{align*}
    \se(3) = \left\{ \begin{pmatrix} \omega^\times & v \\ 0 & 0 \end{pmatrix} \in \R^{4\times 4} : \omega^\times \in \so(3), v \in \R^3 \right\},
\end{align*}
equipped with a Lie bracket given by the matrix commutator, and where $\so(3)$ is the set of real $3 \times 3$ skew-symmetric matrices
\begin{align*}
    \so(3) = \left\{ U \in \R^{3 \times 3} : U + U^\top = 0 \right\}.
\end{align*}

Let $(\omega, v) \in \R^6$, then the wedge operator
\begin{align*}
(\omega, v)^\wedge = \begin{pmatrix} \omega^\times & v \\ 0 & 0 \end{pmatrix} \in \se(3)
\end{align*}
identifies $\R^6$ with $\se(3)$, where
\begin{align*}
    \omega^\times
    = \begin{pmatrix}
        \omega_1 \\ \omega_2 \\ \omega_3
    \end{pmatrix}^\times
    = \begin{pmatrix}
        0 & -\omega_3 & \omega_2 \\
        \omega_3 & 0 & -\omega_1 \\
        -\omega_2 & \omega_1 & 0
    \end{pmatrix} \in \so(3).
\end{align*}
Equipping $\R^6$ with the Lie bracket $[\cdot, \cdot]: \R^6 \times \R^6 \to \R^6$ defined by
\begin{align*}
    \left[ \begin{pmatrix}
    \omega_1 \\ v_1
    \end{pmatrix}, \begin{pmatrix}
    \omega_2 \\ v_2
    \end{pmatrix} \right] = \begin{pmatrix}
    \omega_1 \times \omega_2 \\ \omega_1 \times v_2 - \omega_2 \times v_1
    \end{pmatrix}
\end{align*}
makes this a Lie algebra isomorphism.
The dual algebra $\se(3)^\ast$ (that is, the set of linear functionals acting on $\se(3)$) may also be identified with $\R^6$ via the standard inner product on $\R^6$:
an element $P \in \se(3)^\ast$ is identified with $(\pi,p) \in \R^6$ by $\langle P, u^\wedge \rangle := \begin{psmallmatrix} \pi \\ p \end{psmallmatrix}^\top u$ for all $u \in \R^6$.
Note that $\langle P, u^\wedge \rangle \not= \tr( P^\top u^\wedge)$ in matrix coordinates.

The big `A' Adjoint $\Ad$ and little 'a' adjoint $\ad$ maps are classical operators defined on the Lie algebra of a Lie group \citep{leeSmoothManifolds2003}. 
Since we write the momentum as a vector in $\R^6$ we will work with the matrix representations of the maps $\Ad^\vee : \SE(3) \times \R^6 \to \R^6$, \hbox{$\ad^\vee : \R^6 \times \R^6 \to \R^6$} \citep{leeSmoothManifolds2003} and their dual maps are given by:
\begin{align*}
    \Ad^\vee_{(R, x)}             & = \begin{pmatrix*}
        R & 0 \\ x^\times R & R
    \end{pmatrix*},  &
    \Ad^{\ast \vee}_{(R, x)}      & =  \begin{pmatrix*}
        R^\top & -R^\top x^\times \\ 0 & R^\top
    \end{pmatrix*}, \\
    \ad^\vee_{(\omega, v)}        & = \begin{pmatrix*}
        \omega^\times & 0 \\ v^\times & \omega^\times
    \end{pmatrix*}, &
    \ad^{\ast \vee}_{(\omega, v)} & = \begin{pmatrix*}
        -\omega^\times & -v^\times \\ 0 & -\omega^\times
    \end{pmatrix*}.
\end{align*}
For $X \in \SE(3)$, $\dot{X} = X U$, the following identities hold:
\begin{align}
    \ddt \Ad^\vee_{X} &= \Ad^\vee_{X}\ad^\vee_{U^\vee} = \ad^\vee_{\Ad^\vee_{X} U^\vee} \Ad^\vee_{X},\\
    \ddt \Ad^{\ast \vee}_{X^{-1}} &= -\Ad^{\ast \vee}_{X^{-1}} \ad^{\ast \vee}_{U} 
    = -\ad^{\ast \vee}_{\Ad^\vee_{X} U^\vee} \Ad^{\ast \vee}_{X^{-1}}. \label{eq:ad_ids}
\end{align}

The matrix $\so(3)$ projection operator $\mathbb{P}_{\so(3)} : \R^{3 \times 3} \to \so(3)$ is defined by
\begin{align}
    \mathbb{P}_{\so(3)}(A) = \frac{1}{2} (A - A^\top) \label{eq:so3_proj}.
\end{align} 
One readily computes that $\tr(B v^\times) = 0$ for any symmetric $B = B^\top \in \R^{3\times 3}$ and that $\tr(u^\times v^\times) = -2 u^\top v$.
Thus, for arbitrary $A \in \R^{3 \times 3}$, one has 
\begin{align}
    \tr(A v^\times) = \tr(\mathbb{P}_{\so(3)}(A) v^\times) = \left[(A - A^\top)^\vee \right]^\top v \label{eq:trace_id},
\end{align} 
where $(\cdot)^\vee : \so(3) \to \R^3$ is the inverse of the $(\cdot)^\times$ operator.

\section{Dynamical model for submersible.}
\label{sec:dynamics_sub}

The modelling approach follows \cite[Section 2.3]{leonardStabilityBottomheavyUnderwater1997}.
We assume that the submersible vehicle considered is fully actuated and designed to be neutrally buoyant (that is, that the forces due to buoyancy and gravity are equal).
Additionally, we assume that the vehicle has been designed so that the centre of buoyancy and the centre of mass coincide. 
The vehicle is modelled as a rigid body subjected to added mass effects from the surrounding fluid that will manifest as additive inertia terms.
In this case, the Lagrangian governing rigid body motion is given by the kinetic energy and may be written as
\begin{align*}
K(\omega, v) =\frac{1}{2} (\omega, v) ^\top \inert (\omega, v),
\end{align*}
where $v$ is the translational (body-frame) velocity, $\omega$ is the rotational (body-frame) velocity and $\inert$ is the coupled moment of inertia.
The inertia $\inert$ is given by
\begin{align*}
\inert &= \begin{pmatrix*}
J & D \\ D^\top  & M
\end{pmatrix*},
\end{align*}
where $M = m \Id + M_{\text{AM}}$ is the sum of the actual and added mass, $J = J_b + J_{AM}$ is the sum of the actual and added rotational inertias, and $D$ is a coupling term that models added mass effects that couple rotational and linear motion.
In components, one defines the angular and linear momentum, respectively, as
\begin{align*}
\pi & \coloneqq J \omega + D v, & p & \coloneqq  Mv + D^\top \omega.
\end{align*}
Define $V = (\omega, v)$ and $P = \inert V$.
Then $K(\omega, v) = K(V)$ can be written as $h(P) = \frac{1}{2} P^\top \inert^{-1} P$.

Then the controlled Lie-Poisson dynamics \citep{marsden1994introduction} on $\se(3)^\ast \cong \R^6$ are given by
\begin{align*}
\begin{pmatrix}
\dot{\pi} \\ \dot{p}
\end{pmatrix} &= \ad^{\ast \vee}_{(\omega, v)}
\begin{pmatrix}
\pi \\ p
\end{pmatrix} + \begin{pmatrix} \tau \\ f \end{pmatrix}\\
&= \begin{pmatrix} -\omega^\times & -v^\times \\ 0 &  -\omega^\times \end{pmatrix} \begin{pmatrix}
\pi \\ p
\end{pmatrix} + \begin{pmatrix} \tau \\ f \end{pmatrix}.
\end{align*}
Writing these in vector form along with the kinematic equations yields
\begin{align}
    \dot{R}   & = R \omega^\times, &\dot{\pi} & = - \omega \times \pi - v \times p + \tau \notag\\
    \dot{x}   & = R v, & \dot{p}   & = - \omega \times p + f \label{eq:se3_dynamics},
\end{align}
which are the Kirchhoff equations of motion for a rigid body in an ideal fluid \citep{kirchhoff1883vorlesungen}.

\section{Lie Group Error and Dynamics}
\label{sec:lie_group_error}

The tracking approach of this paper relies on placing a Lie group structure on the direct product $\SE(3) \times \R^6 \cong \SE(3) \times \so(3)^\ast$ and then formulating the right-invariant error. 
Let $Q \coloneqq (R, x) \in \SE(3)$ denote the configuration and $P \coloneqq (\pi, p) \in \R^6$ denote the total momentum.
The $\SE(3) \ltimes \R^6$ cotangent group structure \citep{marsdenSemidirectProductsReduction1984,hampseyEquivariantTrackingControl2025} is given by the group multiplication
\begin{align}
    (Q_1, P_2)(Q_2, P_2) := (Q_1 Q_2, \Ad_{Q_2}^{\ast \vee} P_1 + P_2). \label{eq:group_mul}
\end{align}
In components, this product takes the form
\begin{align*}
((R_1, x_1), (\pi_1,  p_1)) & ((R_2, x_2),(\pi_2,p_2))\\
& = ((R_1 R_2, R_1 x_2 + x_1), \\
&\qquad (R_2^\top (\pi_1 - x_2 ^\times p_1) + \pi_2, R_2^\top p_1 + p_2)).
\end{align*}
The inverse is given by
\begin{align*}
    (Q, P)^{-1} := (Q^{-1}, -\Ad_{Q^{-1}}^{\ast \vee} P).
\end{align*}
In components one has 
\begin{align*}
((R, x), (\pi, p))^{-1} = ((R^\top, -R^\top x), (-R \pi - x \times R  p, -R p )).
\end{align*}

Consider a desired trajectory 
\begin{align*}
(Q_d, P_d) = ((R_d, x_d), (\pi_d, p_d)) : \R \to \SE(3) \times \R^3,    
\end{align*}
that satisfies the dynamics \eqref{eq:se3_dynamics}.
Utilising the group structure \eqref{eq:group_mul} and the right-invariant error definition leads to 
\begin{align}
    (Q_E, P_E) &\coloneqq (Q, P) (Q_d, P_d^{-1}) \notag\\
               &= (Q Q^{-1}_d, \Ad^{\ast \vee}_{Q^{-1}_d}(P - P_d) ) \label{eq:err_def}.
\end{align}
In components one has 
\begin{align}
R_E &= R R^\top_d, & & \pi_E = R_d (\pi - \pi_d) + x_d \times R_d (p - p_d) \notag\\
x_E &= x - R_E x_d , & & p_E = R_d(p - p_d) \label{eq:err_states}.
\end{align}
    
The choice of the right-invariant error is motivated by the error kinematics.
Letting $U \coloneqq (\omega, v)$, compute
\begin{align*}
    \dot{Q}_E &= \dot{Q} Q^{-1}_d - Q Q^{-1}_d \dot{Q}_d Q^{-1}_d = Q(U - U_d)^\wedge Q^{-1}_d\\
              &= Q_E \Ad_{Q_d} (U - U_d)^\wedge = Q_E U_E,
\end{align*}
where $U_E \coloneqq \Ad^\vee_{Q_d} (U - U_d)$.
Importantly, note that because $P_E \coloneqq \Ad^{\ast \vee}_{Q^{-1}_d}(P - P_d)$, then $P_E = \Ad^{\ast \vee}_{Q^{-1}_d} \inert \Ad^{\vee}_{Q^{-1}_d} U_E$ and so $P_E$ and $U_E$ are canonically conjugate with respect to the ``error Hamiltonian'' 
\begin{align}
    H_E \coloneqq \frac{1}{2} P_E^\top \binert^{-1} P_E  = \frac{1}{2} P_E^\top U_E\label{eq:error_hamiltonian},
\end{align}
where 
\begin{align}
    \binert \coloneqq \Ad^{\ast \vee}_{Q^{-1}_d} \inert \Ad^{\vee}_{Q^{-1}_d}. \label{eq:binert_definition}
\end{align}
That is, $U_E = \pp{H_E}{P_E}$.
This construction motivates the following definitions:
\begin{align}
\omega_E &\coloneqq R_d (\omega - \omega_d), \notag\\
v_E &\coloneqq x_d \times \omega_E + R_d (v - v_d),\notag\\
\tau_E &\coloneqq - \omega_E \times R_d \pi_d - (R_d (v - v_d)) \times R_d p_d + R_d (\tau - \tau_d),\notag\\
f_E &\coloneqq -\omega_E \times R_d p_d + R_d (f - f_d). \label{eq:err_inputs}
\end{align}

Then, from Equation \eqref{eq:err_def}, the dynamics of the error momentum $P_E$ are given by
\begin{align}
\dot{P}_E &= (\ddt \Ad^{\ast \vee}_{Q^{-1}_d})(P - P_d) + \Ad^{\ast \vee}_{Q^{-1}_d}(\dot{P} - \dot{P}_d)\notag \\
&= (-\ad^{\ast \vee}_{\Ad^{\vee}_{Q_d(t) U_d(t)}} \Ad^{\ast \vee}_{Q_d^{-1}(t)})(P - P_d)\notag \\
&\qquad  + \Ad^{\ast \vee}_{Q^{-1}_d}(\ad^{\ast \vee}_U P + \tau - \ad^{\ast \vee}_{U_d} P_d - \tau_d)\notag \\
&\qquad + \Ad^{\ast \vee}_{Q^{-1}_d}(\tau - \tau_d)\notag \\
&=  \ad^{\ast \vee}_{\Ad^{\vee}_{Q_d} (U - U_d)} \Ad^\ast_{Q^{-1}_d} P + \Ad^{\ast \vee}_{Q^{-1}_d}(\tau - \tau_d)\notag \\
&=  \ad^{\ast \vee}_{U_E} P_E + \ad^{\ast \vee}_{U_E} \Ad^{\ast \vee}_{Q^{-1}_d} P_d + \Ad^{\ast \vee}_{Q^{-1}_d}(\tau - \tau_d)\notag\\
&=  \ad^{\ast \vee}_{U_E} P_E + \tau_E \label{eq:P_E_dot},
\end{align}
where $\tau_E$ is defined by 
\begin{align*}
    \tau_E \coloneqq \ad^{\ast \vee}_{U_E} \Ad^{\ast \vee}_{Q^{-1}_d} P_d + \Ad^{\ast \vee}_{Q^{-1}_d}(\tau - \tau_d).
\end{align*}
Note that \eqref{eq:P_E_dot} is the Euler-Poincare dynamics in error coordinates and the error system is Hamiltonian. 
In components, these dynamics are given by:
\begin{align}
    \dot{R}_E &= R_E \omega_E^\times & &     \dot{\pi}_E = - \omega_E \times \pi_E  - v_E \times p_E + \tau_E \notag\\
    \dot{x}_E &= R_E v_E & &\dot{p}_E = -\omega_E \times p_E + f_E \label{eq:err_dynamics}.
\end{align}

The canonical relationship $U_E = \pp{H_E}{P_E}$, that underlies the Hamiltonian structure of the error dynamics, does come at the cost that the new inertia $\binert$ \eqref{eq:error_hamiltonian} is time varying. 
It is through this structure that the energy supplied to the error system through the evolution of the general desired trajectory $(Q_d, P_d)$ is modelled in the Hamiltonian error system.

\section{Tracking Control}
\label{sec:tracking_control}

In this section, we propose a control law that stabilises the error system to identity.
The design is based on applying energy shaping and damping injection in the Hamiltonian error coordinates. 

\begin{theorem}
Let $\xi_d \coloneqq ((R_d, x_d), (\pi_d, p_d))  : \R \to \SE(3) \times \R^6$ be a desired trajectory with input $(\tau_d, f_d): \R \to \R^6$, and let $\xi \coloneqq ((R, x), (\pi, p)) : \R \to \SE(3) \times \R^6$ be a system trajectory.
Both trajectories satisfy the dynamics \eqref{eq:se3_dynamics}.
The error $\xi_E \coloneqq ((R_E, x_E), (\pi_E, p_E))$ is given by \eqref{eq:err_states}.   
Let $K^R_x \in \R^{3 \times 3}$ be a diagonal symmetric matrix whose eigenvalues satisfy\footnote{The eigenvalue condition ensures the Hessian of $\tr(K^R_x (\Id - R_E))$ is positive definite on $\SO(3)$.}
 $\lambda_i + \lambda_j > 0$ for $i \neq j \in \{1, 2, 3\}$, and let $K^R_d, K^x_p$ and $K^x_d \in \R^{3\times3}$ be symmetric positive-definite matrices.
Let the inputs $\tau, f$ be given by
\begin{subequations}
\begin{align}
    \tau &:= \tau_d + \tilde{\omega} \times \pi_d + \omega_d \times \tilde{\pi} + \tilde{v} \times p_d + v_d \times \tilde{p} \notag\\
    & \quad\quad - R_d^\top K^R_d \omega_E - R_d^\top \mathbb{P}_{\so(3)}\left(K^R_p R_E\right)^\vee \notag\\
    & \quad\quad\quad\quad + (R_d^\top x_d) \times (R_d^\top K^x_d v_E + R^\top K^x_p x_E), \label{eq:tau_input}\\
    f &:= f_d + \tilde{\omega} \times p_d + \omega_d \times \tilde{p} - R_d^\top K^x_d v_E - R^\top K^x_p x_E \label{eq:f_input}.   
\end{align}
\end{subequations}
Then $\xi_E \to ((\Id, 0), (0, 0))$ asymptotically.
Furthermore, if $x_d$ is bounded then $\xi \to \xi_d$ asymptotically.
\end{theorem}
\begin{proof}

Consider the candidate Lyapunov function given by
\begin{align*}
    \calL_E &= H_E +  \tr(K^R_p (\Id - R_E)) +  x_E^\top K^x_p  x_E\\
            &= \frac{1}{2}\pi_E^\top \omega_E + \frac{1}{2} p_E^\top v_E +  \tr(K^R_p (\Id - R_E)) +  x_E^\top K^x_p  x_E,
\end{align*}
where $H_E \coloneqq \frac{1}{2}\pi_E^\top \omega_E + \frac{1}{2} p_E^\top v_E$ is the ``error kinetic energy'' defined in \eqref{eq:error_hamiltonian} and $\tr(K_R(\Id - R_E)) +  x_E^\top K_x x_E$ is a potential energy term with minimum at $R_E = \Id, x_E = 0$.
We will consider the time derivatives of the kinetic energy and potential energy terms separately.  
Taking the time derivative of $H_E$,
\begin{align*}
    \dot{H}_E &= \ddt \frac{1}{2} P_E^\top \binert^{-1} P_E\\
    &= P_E^\top \binert^{-1} \dot{P}_E + \frac{1}{2}P_E^\top (\ddt \binert^{-1}) P_E\\
    &= U_E^\top (\ad^\ast_{U_E} P_E + \tau_E) + \frac{1}{2}P_E^\top (\ddt \binert^{-1}) P_E\\
    &= U_E^\top \tau_E + \frac{1}{2}P_E^\top (\ddt \binert^{-1}) P_E.
\end{align*}

Using the definition
$\binert \coloneqq \Ad^{\ast \vee}_{Q^{-1}_d} \inert \Ad^{\vee}_{Q^{-1}_d}$ \eqref{eq:binert_definition} and the identities \eqref{eq:ad_ids}, one writes 
\begin{align*}
    \ddt \binert^{-1} &= (\ddt \Ad^{\vee}_{Q_d}) \inert^{-1} \Ad^{\ast \vee}_{Q_d} + \Ad^{\vee}_{Q_d} \inert^{-1} (\ddt \Ad^{\ast \vee}_{Q_d})\\
    &= \ad^{\vee}_{\Ad_{Q_d V_d}} \Ad^{\vee}_{Q_d} \inert^{-1} \Ad^{\ast \vee}_{Q_d}\\
    &\qquad + \Ad^{\vee}_{Q_d} \inert^{-1}  \Ad^{\ast \vee}_{Q_d} \ad^{\ast \vee}_{\Ad^{\vee}_{Q_d V_d}}\\
&= \ad^{\vee}_{\Ad^{\vee}_{Q_d V_d}} \binert^{-1} + \binert^{-1} \ad^{\ast \vee}_{\Ad^{\vee}_{Q_d V_d}},
\end{align*}
and hence 
\begin{align*}
    P_E^\top (\ddt \binert^{-1}) P_E &= P_E^\top (\ad^{\vee}_{\Ad^{\vee}_{Q_d V_d}} \binert^{-1} + \binert^{-1} \ad^{\ast \vee}_{\Ad^{\vee}_{Q_d V_d}}) P_E\\
&= 2 P_E^\top \ad^{\vee}_{\Ad^{\vee}_{Q_d V_d}} U_E.
\end{align*}
Thus, one has 
\begin{align*}
    \dot{H}_E &= V^\top_E F_E + P_E^\top \ad^{\vee}_{\Ad^{\vee}_{Q_d V_d}} U_E\\
    &= \omega^\top_E \tau_E + v^\top_E f_E + \pi^\top_E R_d \omega^\times_d R_d^\top \omega_E\\
&\qquad + p^\top_E ((x_d \times R_d \omega_d) \times \omega_E)\\
&\qquad + p^\top_E R_d v^\times_d R^\top_d \omega_E + p^\top_E R_d \omega^\times_d R^\top_d v_E.
\end{align*}
Expanding the definitions of the error terms yields the rather unwieldy expression
\begin{align*}
    \dot{H}_E &= \omega^\top_E (- \omega_E \times R_d \pi_d)\\
    & \qquad + \omega^\top_E [-(R_d (v - v_d)) \times R_d p_d + R_d (\tau - \tau_d)]\\
    & \qquad + v^\top_E (-\omega_E \times R_d p_d + R_d (f - f_d))\\
    &\qquad + \pi^\top_E R_d \omega^\times_d R_d^\top \omega_E  + p^\top_E ((x_d \times R_d \omega_d) \times \omega_E)\\
&\qquad + p^\top_E R_d v^\times_d R^\top_d \omega_E + p^\top_E R_d \omega^\times_d R^\top_d v_E\\
&= \omega^\top_E [-(R_d (v - v_d)) \times R_d p_d + R_d (\tau - \tau_d)]\\
    &\qquad + v^\top_E (-\omega_E \times R_d p_d + R_d (f - f_d))\\
&\qquad + \pi^\top_E R_d \omega^\times_d R_d^\top \omega_E + p^\top_E ((x_d \times R_d \omega_d) \times \omega_E)\\
&\qquad + p^\top_E R_d v^\times_d R^\top_d \omega_E + p^\top_E R_d \omega^\times_d R^\top_d v_E\\
&= \tilde{\omega}^\top [-\tilde{v} \times p_d + \tilde{\tau}]\\
    &\qquad + ((R^\top_d x_d) \times \tilde{\omega} + \tilde{v})^\top (-\tilde{\omega} \times p_d + \tilde{f})\\
&\qquad + (\tilde{\pi} + (R^\top_d x_d) \times \tilde{p})^\top \omega_d \times \tilde{\omega}\\
&\qquad + \tilde{p}^\top (((R^\top_d x_d) \times \omega_d) \times \tilde{\omega})\\
&\qquad + \tilde{p}^\top v^\times_d \tilde{\omega} + \tilde{p}^\top \omega^\times_d ((R^\top_d x_d) \times \tilde{\omega} + \tilde{v}).
\end{align*}

Substituting in the control laws \eqref{eq:tau_input} and \eqref{eq:f_input} one computes 
\begin{align}
\dot{H}_E &= \tilde{\omega}^\top [-\tilde{v} \times p_d +  \tilde{\omega} \times \pi_d + \omega_d \times \tilde{\pi} + \tilde{v} \times p_d + v_d \times \tilde{p} \notag\\
    & \qquad - R_d^\top K^R_d \omega_E - R_d^\top \mathbb{P}_{\so(3)}\left(K^R_p R_E\right)^\vee \notag\\
    & \qquad + (R_d^\top x_d) \times (R_d^\top K^x_d v_E + R^\top K^x_p x_E)]\notag\\
    &\quad + ((R^\top_d x_d) \times \tilde{\omega} + \tilde{v})^\top [-\tilde{\omega} \times p_d + \tilde{\omega} \times p_d + \omega_d \times \tilde{p}\notag\\
    &\qquad \qquad - R_d^\top K^x_d v_E - R^\top K^x_p x_E]\notag\\
&\quad + (\tilde{\pi} + (R^\top_d x_d) \times \tilde{p})^\top \omega_d \times \tilde{\omega}\notag\\
&\quad + \tilde{p}^\top [((R^\top_d x_d) \times \omega_d) \times \tilde{\omega}]\notag\\
&\quad + \tilde{p}^\top v^\times_d \tilde{\omega}  + \tilde{p}^\top \omega^\times_d ((R^\top_d x_d) \times \tilde{\omega} + \tilde{v}) \notag\\
&= - \omega^\top_E K^R_d \omega_E - v^\top_E K^x_d v_E\notag\\
&\qquad - \omega^\top_E \mathbb{P}_{\so(3)}\left(K^R_p R_E\right)^\vee - v^\top_E R_E^\top K^x_p x_E. \label{eq:dot_H_id}
\end{align}
That is, the change in kinetic energy comprises two dissipation terms and two coupling terms.
Note that the time derivative of the potential terms in $\calL_E$ are 
\begin{align}
    \ddt &\left(\tr(K^R_p (\Id - R_E)) +  \frac{1}{2}x_E^\top K^x_p  x_E \right) \notag\\
    &= -\tr(K^R_p  R_E \omega^\times_E) + x_E^\top K^x_p  R_E v_E \notag\\
    &= \omega^\top_E \mathbb{P}_{\so(3)}\left(K^R_p  R_E\right)^\vee+ x_E^\top K^x_p  R_E v_E \label{eq:dot_potential},
\end{align}
where the last line follows from the identity \eqref{eq:trace_id}.
Adding together expressions \eqref{eq:dot_H_id} and \eqref{eq:dot_potential}, yields 
\begin{align*}
    \dot{\calL}_E &= - \omega^\top_E K^R_d \omega_E - v^\top_E K^x_d v_E. 
\end{align*}
It follows that $\calL_E$ is monotonically non-increasing along trajectories of the system. 

The inertia $\binert$ is uniformly bounded above and below, so by the Lasalle-Yoshizawa theorem \citep[Theorem 8.4]{khalilNonlinearSystems2002}, there exists a compact set $\calB$ containing the equilibrium point $((\Id, 0), (0,0))$, such that if $\xi_E(0) \in \calB$, then $\xi_E(t)$ exists for all time, $\xi_E(t) \in \calB$ and $\dot{\calL} \to 0$.
Thus, $\omega_E\ \to 0$, $v_E \to 0$ and hence $\pi_E \to 0$, $p_E \to 0$.
By Barbalat's Lemma \citep[Lemma 8.2]{khalilNonlinearSystems2002}, this implies that $\dot{\pi}_E \to 0$,  $\dot{p}_E \to 0$, and from the error dynamics \eqref{eq:err_dynamics}, the definition of the input errors \eqref{eq:err_inputs} and the closed loop inputs \eqref{eq:f_input}, \eqref{eq:tau_input}, this implies $\mathbb{P}_{\so(3)}(K^R_p R_E) \to 0$ and $x_E \to 0$.
In particular, $\xi_E = ((\Id, 0), (0, 0))$ is an isolated equilibrium point, so there exists some neighborhood of $((\Id, 0), (0, 0))$ for which $\xi_E \to ((\Id, 0), (0, 0))$ asymptotically. 

To show that $R \to R_d$, note that because the Lie group $\SO(3)$ is compact, there exists some upper bound $L$ such that $\| R_d \| < L$ for all $R_d \in \SO(3)$.
Then $\| R - R_d \| = \| (R R^\top_d - \Id) R_d \| \leq L \| R_E - \Id \|$, so for a given $\varepsilon > 0$, if $\| R_E - \Id \| < \frac{1}{L} \varepsilon$,then $\| R - R_d \| < \varepsilon$.
Similarly, to show that $x \to x_d$ if $x_d$ is bounded, let $\|x_d \| < M$ and note that by the triangle inequality, we have
\begin{align*}
\|x - x_d \| &\leq \| x - R_E x_d \| + \| R_E x_d - x_d \|\\
& \leq \| x_E \| + M \|R_E - \Id \|.    
\end{align*}

Thus, $x$ can be made arbitrarily close to $x_d$ by making $x_E$ sufficiently close to $0$ and $R_E$ sufficiently close to $\Id$.
Similar arguments hold for $\pi \to \pi_d$ and $p \to p_d$.
\hfill$\Box$
\end{proof}

\begin{remark}
\label{remark:pd_law}
The control law \eqref{eq:tau_input}, \eqref{eq:f_input} is simply a geometric PD-control law formed in the error coordinates described in \eqref{eq:err_inputs}, with an additional term that compensates for the time variation in $\binert^{-1}$.
\end{remark}

%%%%%%%%%%%%%%%%%%%%%%%%%%%%%%%%%%%%%%%%%%%%%%%%%%%%%%%%%%%%%%%%%%%%%%%%%%%%%%
\section{Simulation}
\label{sec:simulation}

To empirically verify the approach, we perform a simulation study in Python.
We simulate a simple tracking task for a BlueROV2 Heavy underwater submersible (Figure \ref{fig:bluerov2}).
\begin{figure}[!tb]
   \includegraphics[width=0.5\linewidth]{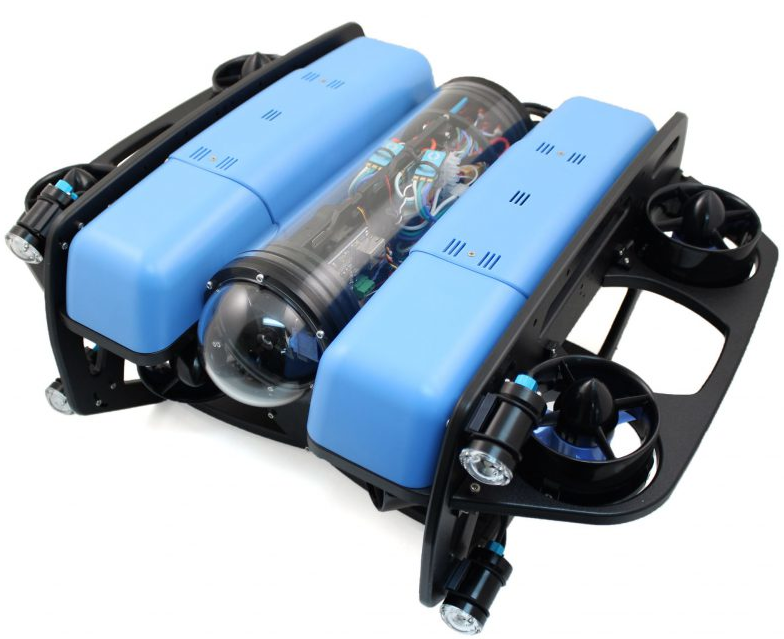}
   \centering
   \caption{BlueROV2 Heavy. Image sourced from \cite{blueroboticsBlueROV2HeavyKit2025}}
   \label{fig:bluerov2}
\end{figure}
The model parameters \citep{wu6DoFModellingControl2018} are given by the following:
\begin{align*}
J &= 0.28 \Id_3, D =  \begin{pmatrix} 0, 0, 0.23 \end{pmatrix}^\times, M = \diag(17, 24.2, 26.07).
\end{align*}

The trajectory is a helix of radius $2$ metres (Figure \ref{fig:traj_3d}b).
For the simulation, the initial position and attitude are perturbed by $x'_0 = x_0 + \begin{pmatrix} 1.0 &  1.5 & 0.8 \end{pmatrix}^\top$ and $R'_0 = \tR_0 \exp \left( \begin{pmatrix} 1.5 & 4.5 & 0.6 \end{pmatrix}^\times \right)$.
The Lyapunov function as a function of time is shown in Figure \ref{fig:traj_3d}a and the convergence of states is shown in Figures \ref{fig:traj_3d}b and \ref{fig:traj_momenta}.
The Lyapunov function is strictly non-increasing until approximately 10 seconds, where some small oscillation is observed - here, the magnitude of the Lyapunov function is of the order $10^{-5}$ and the authors believe this oscillation is due to numerical instability.

\begin{figure}[!tb]
\centering
    \subfloat[]{\includegraphics[width=0.47\linewidth]{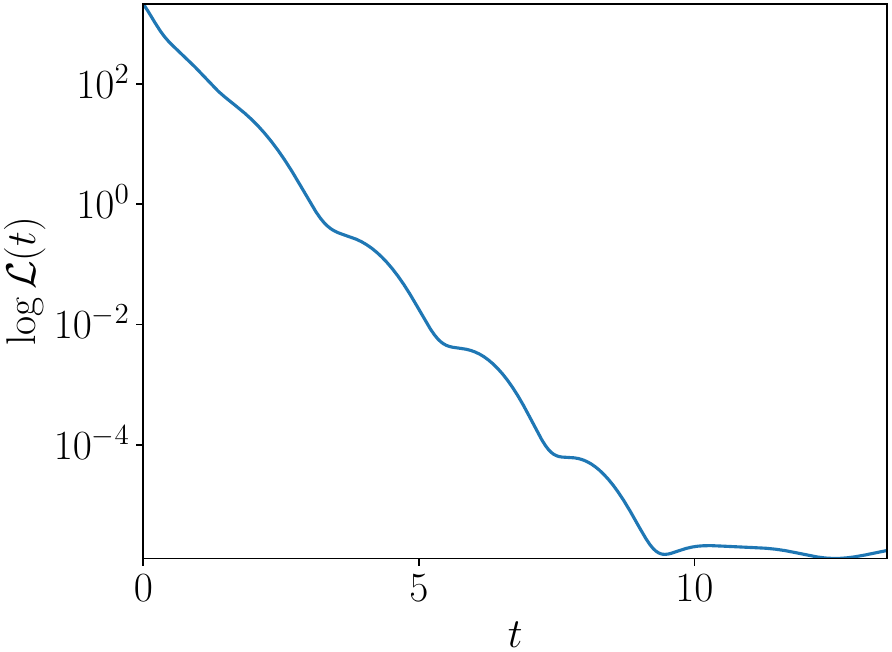}}%
    \quad
    \subfloat[]{\includegraphics[width=0.47\linewidth]{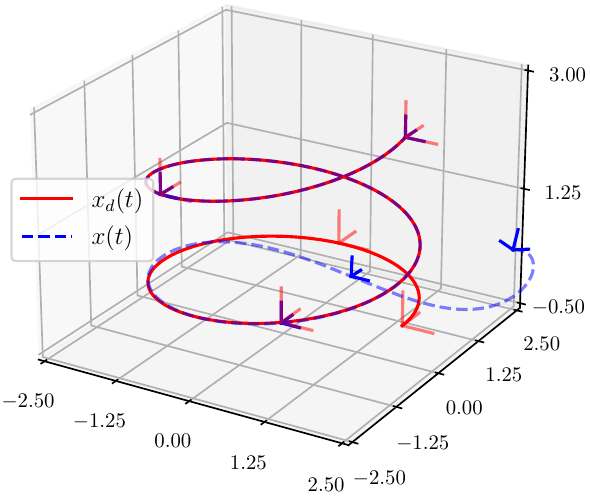}}%
   \caption{(a): The log Lyapunov function $\log \mathcal{L}(t)$ for the helix tracking simulation. (b): Desired position $x_d$ (red) vs actual position $x$ (blue-dashed) trajectories for $\SE(3)$ helix tracking simulation. The attitudes $R_d$ and $R$ are represented by the orthogonal frames superimposed on the trajectory.}
   \label{fig:traj_3d}
\end{figure}
\begin{figure}[!tb]
\centering
    \subfloat[]{\includegraphics[width=0.47\linewidth]{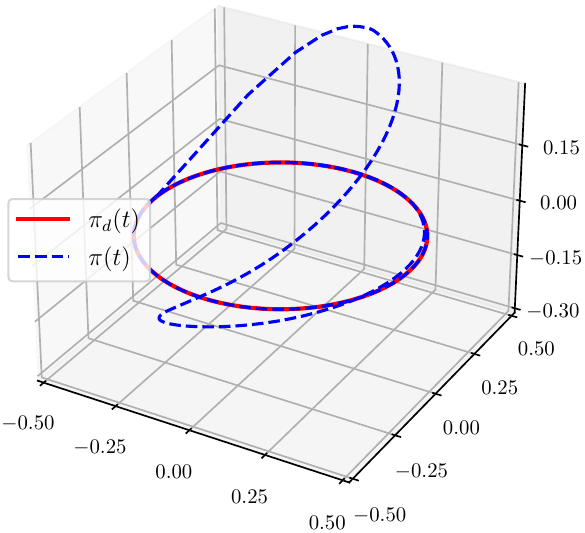}}%
    \quad
    \subfloat[]{\includegraphics[width=0.47\linewidth]{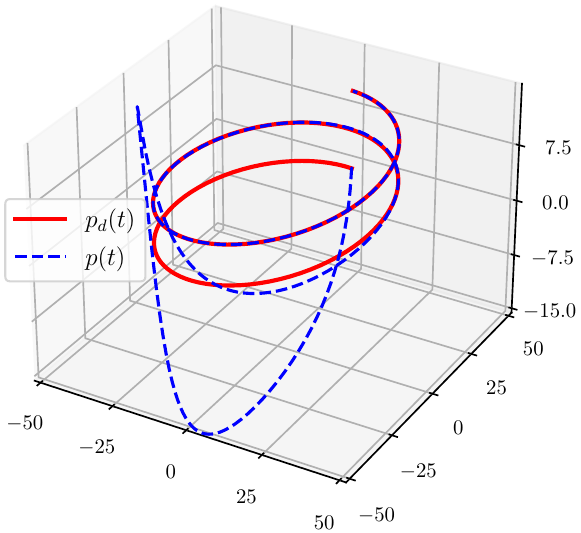}}%
   \caption{(a): Desired angular momentum $\pi_d$ (red) vs actual angular momentum $\pi$ (blue-dashed) trajectories for $\SE(3)$ helix tracking simulation. (b): Desired momentum $p_d$ (red) vs actual angular momentum $p$ (blue-dashed) trajectories for $\SE(3)$ helix tracking simulation.}
   \label{fig:traj_momenta}
\end{figure}

\section{Discussion}
\label{sec:discussion}
It is clear from Figures \ref{fig:traj_3d} and \ref{fig:traj_momenta} that the control causes the perturbed system to correctly track the desired trajectory.
The main benefit to our approach is the ability to treat the error system as a de facto Hamiltonian system, leveraging its natural passivity properties.
This simplifies the control design significantly, as it can be treated as a PD control in the error coordinates (see Remark \ref{remark:pd_law})
However, it is not without downsides - the main one being that the right-invariant Lie group error requires the introduction of an inertial frame.
This can be clearly seen in \eqref{eq:err_states}, where the velocity and input terms are inertial frame quantities.
This introduces a dependence on the desired frame into the control, cf.~\eqref{eq:tau_input}, which manifests as a lever arm coupling together the rotational and translational control terms.
Future work includes exploring this further and investigating ways to mitigate the effect of these terms in the closed-loop control. 
Additionally, the authors plan to consider relaxation of the assumption of coincident centres of gravity and buoyancy by exploiting advected parameter modelling \citep{contrerasControlledLagrangiansStabilization2022}. 

\section{Conclusion}
In this paper, we have addressed the tracking control problem for underwater vehicles on $\SE(3)$ with coincident centres of mass and buoyancy.
To do this, we extended the natural $\SE(3)$ symmetry of the configuration space to a semi-direct product symmetry on the phase space $\SE(3) \times \R^6$, and used the natural Lie group error of this symmetry to define an equivariant error. 
We showed that, with a transformation of the input, the resultant error system was itself a Lie-Poisson system, which motivates a natural Lyapunov function definition using the energy of the error system.
We provided a stabilising control for this system, proved asymptotic stability, and verified the approach in simulation, with simulation parameters taken from a real submersible.

\section*{Acknowledgement}
This research was supported by the Australian Research Council through Discovery Grant DP210102607 `` Exploiting the Symmetry of Spatial Awareness for 21st Century Automation''.

\bibliographystyle{plainnat}        % Include this if you use bibtex
\bibliography{references_archived}             % bib file to produce the bibliography
                                                     % with bibtex (preferred)
                                                   
%\begin{thebibliography}{xx}  % you can also add the bibliography by hand

%\bibitem[Able(1956)]{Abl:56}
%B.C. Able.
%\newblock Nucleic acid content of microscope.
%\newblock \emph{Nature}, 135:\penalty0 7--9, 1956.

%\bibitem[Able et~al.(1954)Able, Tagg, and Rush]{AbTaRu:54}
%B.C. Able, R.A. Tagg, and M.~Rush.
%\newblock Enzyme-catalyzed cellular transanimations.
%\newblock In A.F. Round, editor, \emph{Advances in Enzymology}, volume~2, pages
%  125--247. Academic Press, New York, 3rd edition, 1954.

%\bibitem[Keohane(1958)]{Keo:58}
%\newblock \emph{Power and Interdependence: World Politics in Transitions}.
%\newblock Little, Brown \& Co., Boston, 1958.

%\bibitem[Powers(1985)]{Pow:85}
%\newblock Is there a way out?
%\newblock \emph{Harpers}, pages 35--47, June 1985.

%\bibitem[Soukhanov(1992)]{Heritage:92}
%\newblock \emph{{The American Heritage. Dictionary of the American Language}}.
%\newblock Houghton Mifflin Company, 1992.

%\end{thebibliography}

\appendix
% \section{A summary of Latin grammar}    % Each appendix must have a short title.
% \section{Some Latin vocabulary}              % Sections and subsections are supported  
%                                                                          % in the appendices.
\end{document}